\documentclass[conference]{IEEEtran}
 \onecolumn
\usepackage{amsmath}
\usepackage{amssymb}
\usepackage{mathrsfs}
\usepackage{tikz}
\usepackage[utf8]{inputenc}
\usepackage{amsfonts} 
\usepackage{epstopdf}
\usepackage{graphicx}
\input{xypic}
\usepackage{bbm}
\usepackage{enumerate}
\usepackage{epsf,epsfig,verbatim,amssymb,amsmath,array,cite,amsthm,subfigure,multicol,multirow}  
\usepackage{psfrag,bm,xspace}
\usepackage{slashbox}
\usepackage{hhline}
\usepackage{mathabx}

\newcommand{\beq}{\begin{equation}}
\newcommand{\eeq}{\end{equation}}

\newtheorem{lem}{Lemma}

\newtheorem{theorem}{Theorem}
\newtheorem{proposition}{Proposition}
\newtheorem{corollary}{Corollary}

\newtheorem{definition}{Definition}

\theoremstyle{definition}
\newtheorem{example}{Example}

\newtheorem*{prob*}{Problem}

\theoremstyle{remark}
\newtheorem{remark}{Remark}

\allowdisplaybreaks 

\global\long\def\ZZ{\mathbb{Z}}
\global\long\def\EE{\mathbb{E}}

\global\long\def\11{\mathbbm{1}}

\allowdisplaybreaks

\begin{document}

\title{ A New Achievable Rate Region for Multiple-Access Channel with States}
\author{
 \IEEEauthorblockN{Mohsen Heidari}
  \IEEEauthorblockA{EECS Department\\University of Michigan\\ Ann Arbor,USA \\
    Email: mohsenhd@umich.edu} 

\and
 \IEEEauthorblockN{Farhad Shirani}
  \IEEEauthorblockA{EECS Department\\University of Michigan\\ Ann Arbor,USA \\
    Email: fshirani@umich.edu }

  \and
  \IEEEauthorblockN{S. Sandeep Pradhan}
  \IEEEauthorblockA{EECS Department\\University of Michigan\\ Ann Arbor,USA \\
    Email: pradhanv@umich.edu}
}
\IEEEoverridecommandlockouts

\maketitle
\begin{abstract}
The problem of reliable communication over the multiple-access channel (MAC) with states is investigated. 
We propose a new coding scheme for this problem which uses \textit{quasi-group} codes (QGC). We derive a new computable single-letter characterization of the achievable rate region. As an example, we investigate 
the problem of \textit{doubly-dirty} MAC with modulo-$4$ addition. It is shown that the sum 
rate $R_1+R_2=1$ bits per channel use is achievable using the new scheme. Whereas, the natural
extension of the Gel'fand-Pinsker scheme, sum-rates greater than $0.32$ are not achievable.    
\end{abstract}
\section{Introduction}
\IEEEPARstart{C}{onsider} reliable communication over a point-to-point channel with channel state available at the transmitter.
Gel'fand and Pinsker introduced a coding strategy for this problem  \cite{Gelfand-Pinsker} which uses \textit{random binning}. It was shown that the capacity is given by  
\begin{align*}
\mathcal{C}=\max_{p(x,u|s)} I(U;Y)-I(U;S).
\end{align*} 
The additive Gaussian channel with state problem was solved by Costa \cite{Costa}. While the point-to-point problem was solved by Gel'fand and Pinsker, characterizing the capacity region of the multiple-access channel (MAC) with non-causal side-information available at the transmitters remains an open problem. One possible coding scheme is the natural extension of the Gel'fand-Pinsker scheme which was introduced in \cite{Jafar-MAC-state}. A well-studied example of the problem of MAC with states is called the \textit{doubly dirty} MAC problem. In this setup, the channel is binary-additive, and the relation between the inputs and the output is as follows: 
\begin{align}\label{eq: doubly dirty MAC}
Y=X_1 \oplus S_1\oplus X_2\oplus S_2,
\end{align}
where $X_1$ is the first encoder's output, and $X_2$ is the second encoder's output. The states $S_1$, and $S_2$ are available at the first and second transmitter, respectively. $S_1$ and $S_2$ are two independent states which are distributed uniformly over $\{0,1\}$.  Each input sequence must satisfy the cost-constraint $\frac{1}{n} \EE\{c_i(X^n_i)\}\leq \tau_i$ for some cost-functions $c_i(\cdot),i=1,2$, as $n\rightarrow \infty$. Philosof and Zamir \cite{Philosof-Zamir} investigated a special case of this problem in which the cost functions are the \textit{Hamming weight}.  They presented a coding scheme which uses linear codes to align the interference. They showed that the natural extension of the Gel'fand-Pinsker scheme is suboptimal. They showed that the capacity region consists of all rate-pairs $(R_1,R_2)$ such that
$$R_1+R_2 \leq \min\{h_b(\tau_1), h_b(\tau_2)\},$$
where $h_b(\cdot)$ is the \textit{binary entropy} function.  The Philosof-Zamir scheme is optimal in this example. However, it highly relies on the additive and symmetric structure of the channel. The scheme is not generalizable to non-additive channels. Later, a coding scheme based on \textit{coset codes} was introduced for the general MAC with states problem \cite{Arun-MAC-with-States}. 
 %
%
%
 In both of these works, schemes using structured codes are used to improve upon the previous known coding schemes which were based on unstructured codes. Similar observations have been made in other multi-terminal problems, such as the K\"orner-Marton source coding problem  \cite{korner-marton}, the joint source-channel coding over MAC \cite{ISIT-paper-MAC-corr-sources}, multiple-descriptions problem \cite{QLC-ISIT16}, and the problem of computation over MAC \cite{Nazer_Gasper_Comp_MAC}.  

In this work, we first consider the \textit{quaternary} additive MAC with states, where all inputs and states are quaternary, and the addition is $\ZZ_4$ addition. In \cite{Arun-MAC-with-States}, group codes are used to derive an achievable region for this example.
 \textit{Group codes} are structured codes which are closed under a group operation. Recently, we introduced a new class of structured codes called \textit{quasi-group} codes (QGC) \cite{QGC-ISIT16}. A QGC is a subset of a group code. Linear codes and group codes are special cases of QGC. QGCs are not necessarily closed under group addition. QGCs span the spectrum from completely structured codes (such as group codes and linear codes) to completely unstructured codes. These codes were used in the K\"orner-Marton problem for modulo prime-power sums \cite{QGC-ISIT16}. For this problem, a coding scheme based on QGCs is presented which strictly improves upon the previously known schemes.  

Next, we propose a new coding strategy using QGCs for the general problem of two-user MAC with independent states. We introduce \textit{nested QGCs}, and propose a binning technique for such codes. A single-letter characterization of the achievable rates  is derived. As an example, we investigate the quaternary doubly dirty MAC. We show that QGCs achieve the sum-rate $R_1+R_2=1$ bits per channel use. Whereas using the natural extension of Gel'fand-Pinsker, sum-rates greater than $0.32$ are not achievable. 

The rest of this paper is as follows: Section II presents the preliminaries and definitions. Section III provides and overview for QGC. Section IV contains the main results of this paper. Section V presents the application of QGC for the doubly-dirty MAC. Finally, Section VI concludes the paper.

\section{Preliminaries and Problem Formulation}
\subsection{Notations}
We denote (i) vectors using lowercase bold letters such as $\mathbf{b}, \mathbf{u}$, (ii) matrices using uppercase bold letters such as $\mathbf{G}$, (iii) random variables using capital letters such as $X,Y$, (iv) numbers, realizations of random variables and elements of sets using lower case letters such as $a, x$. Calligraphic letters such as $\mathcal{C}$ and  $\mathcal{U}$  are used to represent sets. 

We denote the set $\{1, 2, \dots, m\}$ by $[1:m]$, where $m$ is an integer. Given a prime power $p^r$, the \textit{ring} of integers modulo $p^r$ is denoted by $\ZZ_{p^r}$. The underlying set is for such group is $\{0,1,\cdots, p^r-1\}$, and the addition and multiplication is modulo-$p^r$. For any $0\leq t \leq r$, denote $H_{t}\triangleq \{t\cdot a: a\in \ZZ_{p^r}\}$. Given $H_t$, any element $a\in \ZZ_{p^r}$ can be uniquely written as $a=h+g$, where $h\in H_t, g\in [0:p^t-1]$. We denote such $g$ by $[a]_{t}$. Given two subsets $\mathcal{U}, \mathcal{V} \subseteq \ZZ_{p^r}^k$, we define a new subset defined as $\{\mathbf{u}\oplus \mathbf{v}: \mathbf{u}\in \mathcal{U}, \mathbf{v}\in \mathcal{V}\}$. We denote such set as $\mathcal{U}\oplus \mathcal{V}$.
 
\subsection{Model}
Consider a two-user discrete memoryless  MAC with input alphabets $\mathcal{X}_1,  \mathcal{X}_2$, and output alphabet $\mathcal{Y}$. The transition probabilities between the input and the output of the channel depends on a pair of random variables $(S_1, S_2)$ which are called states. Each state $S_i$ take values from the set $\mathcal{S}_i$, where $i=1,2$. The sequences of the states are independently and identically distributed (i.i.d) according to the probability distribution $p(s_1,s_2)$. Prior to any transmission, the entire sequence of the state $S_i$ is known at the $i$th transmitter, $i=1,2$. The conditional distribution of $Y$ given the inputs and the states is denoted by  $p(y|x_1x_2s_1s_2)$.  Let $y^n$ be the output of the channel after $n$ uses. If $x^n_i$ is the input sequence, and $s_i^n$ is the state sequence, then the following condition is satisfied:
\begin{align*}
p(\mathbf{y}_n|\mathbf{y}^{n-1}, \underline{\mathbf{x}}^{n-1}, \underline{\mathbf{s}}^{n-1})=p(y_n|\underline{x}_n, \underline{s}_n ).
\end{align*}
Each input $X_i$  is associated with a cost function $c_i:\mathcal{X}_i \times \mathcal{S}_i \rightarrow [0, +\infty)$. The input sequence $\mathbf{X}_i^n$ is then constrained to the average cost defined by 
\begin{align*}
\bar{c}_i(\mathbf{X}_i^n, \mathbf{S}_i^n)\triangleq \frac{1}{n}\sum_{j=1}^n c_i(X_{ij}, S_{ij}).
\end{align*}

\begin{definition}
An $(n, \Theta_1, \Theta_2)$-code for reliable communication over a given MAC with states is defined by two encoding functions
$f_i: \{1, 2, \dots, \Theta_i\} \times \mathcal{S}^n_i \rightarrow \mathcal{Y}^n, \quad i=1,2,$
and a decoding function 
$g: \mathcal{Y}^n\rightarrow \{1, 2, \dots, \Theta_1\} \times \{1, 2, \dots, \Theta_2\} .$
\end{definition}

\begin{definition}\label{def: MAC with state achievable rate}
For a given MAC with states, the rate-cost $(R_1,R_2, \tau_1, \tau_2)$ is said to be achievable, if for any $\epsilon >0$, there exist a $(n, \Theta_1, \Theta_2)$-code such that  
\begin{align*}
&P\{g(Y^n)\neq (M_1,M_2)\}\leq \epsilon, \quad 
\frac{1}{n}\log \Theta_i \geq R_i-\epsilon, \quad \EE\{\bar{c}_i (f_i(M_i), \mathbf{S}_i^n)\} \leq \tau_i+\epsilon
\end{align*}
for $i=1,2$, where a)$M_1,M_2$ are independent random variables with distribution $p(M_i=m_i)=\frac{1}{\Theta_i}$ for all $ m_i \in [1:\Theta_i]$, b) $M_i$ is independent of the states $S_1,S_2$. Given $\tau_1, \tau_2$, the capacity region $\mathcal{C}_{\tau_1, \tau_2}$ is defined as the set of all rates $(R_1,R_2)$ such that the rate-cost $(R_1,R_2, \tau_1, \tau_2)$ is achievable.
\end{definition}

\subsection{The Extension of Gel'fand-Pinsker Scheme}
Jafar \cite{Jafar-MAC-state} introduced a natural extension of the Gel'fand-Pincker scheme for the problem of MAC with states, and derived a new achievable rate region using such scheme.  

\begin{proposition}[{\cite{Jafar-MAC-state}}]\label{prep: Gelfand-pinker achievable}
For a MAC $(\mathcal{X}_1,\mathcal{X}_2, \mathcal{Y}, P_{Y|X_1X_2})$ with independent states $(S_1,S_2)$ and cost functions $c_1,c_2$, the \textit{closure} and \textit{convex hull} of all rate-pairs $(R_1,R_2)$ satisfying the following conditions are achievable.
\begin{align}\nonumber
R_1 &\leq I(U_1; Y|U_2Q)-I(U_1;S_1Q)\\\nonumber
R_2 &\leq I(U_2; Y|U_1Q)-I(U_2;S_2|Q)\\
R_1+R_2&\leq I(U_1 U_2; Y|Q)-I(U_1;S_1|Q)-I(U_2;S_2|Q),\label{eq: sum-rate}
\end{align}
where $\EE\{c_i(X_i,S_i)\}\leq \tau_i, i=1,2$, and the joint PMF of all the random variables in the above factors as   $$ p(q)p(s_1)p(s_2) \prod_{i=1,2} p(u_ix_i|s_i q) p(y|x_1x_2).$$
\end{proposition}
To the best of our knowledge, the above rate region is the current largest achievable rate region using unstructured codes for the problem of MAC with states.

\section{An Overview of Quasi Group Codes}
%
We use a class of structured codes called quasi group codes. In this section, we state the definition and key properties of QGCs given in  \cite{QGC-ISIT16}.  

A QGC is defined as a subset of a group code. Such codes are a general form of linear codes and \textit{group codes}. 
Consider a $k\times n$ matrix $\mathbf{G}$ and a $n$-length vector $\mathbf{b}$ with elements in $\ZZ_{p^r}$. Let $\mathcal{U}$ be a subset of $\ZZ_{p^r}^k$. A QGC on $\ZZ_{p^r}$ is defined as 
\begin{align}\label{eq: QGC codebook}
\mathcal{C}=\{\mathbf{u}\mathbf{G}+\mathbf{b}: \mathbf{u}\in \mathcal{U}\}.
\end{align}
 For a general subset $\mathcal{U}$,  it is difficult to derive achievable rates of QGCs using single-letter characterizations. Therefore, we present an special construction of $\mathcal{U}$ for which single-letter characterizations is possible.

Given a positive integer $m$, consider $m$ mutually independent random variables $U_1, U_2, \cdots, U_m$. Suppose each $U_i$ takes values from $\ZZ_{p^r}$ with distribution $p_i(u_i)$. Consider positive integers $k_i, i\in [1:m]$. For $\epsilon>0$, let $A_{\epsilon}^{k_i}(U_i)$ be the collection of all $\epsilon$-typical sequences of $U_i$ with length $k_i$, where $i\in[1:m]$.  Define $\mathcal{U}$ as the Cartesian product of the typical sets of $U_i, i\in [1:m]$, more precisely 
\begin{align}\label{eq: set U cartesian product}
\mathcal{U}\triangleq \bigotimes_{i=1}^m A_{\epsilon}^{(k_i)}(U_i).
\end{align}

For more convenience, we use a notation for this construction. Let $k\triangleq \sum_{i=1}^m k_i$. Denote $q_i \triangleq \frac{k_i}{k}$. Note that $q_i\geq 0$ and $\sum_i q_i=1$. Therefore, we can define a random variable $Q$ with $P(Q=i)=q_i$. Define a random variable $U$ with the conditional distribution $P(U=a|Q=i)=P(U_i=a)$ for all $a\in \ZZ_{p^r}, i\in [1:m]$. With this notation, the set $\mathcal{U}$ in (\ref{eq: set U cartesian product})	 is characterized by $\epsilon, k$ and the pair $( U,Q)$.  Note that for large enough $k$, we have,
$$\frac{1}{n}\log_2 |\mathcal{U}|\approx \frac{k}{n}\sum_{i=1}^m q_i H(U_i)=\frac{k}{n}H(U|Q).$$
\begin{definition}\label{def: QGC}
A  $(n,  k)$-QGC over $\ZZ_{p^r}$ is defined as in (\ref{eq: QGC codebook}), and is characterized by a matrix $\mathbf{G}\in \ZZ_{p^r}^{k\times n}$, a translation $\mathbf{b}\in \ZZ^n_{p^r}$, and a pair of random variables $( U, Q)$ distributed over a finite set $\ZZ_{p^r} \times \mathcal{Q}$. 
\end{definition} 

Let $\mathcal{C}$ be a $(n,k)$-QGC with random variables $(Q,U)$. Suppose the generator matrices and the translation vector of $\mathcal{C}$ are chosen randomly and uniformly from $\ZZ_{p^r}$.  Then for large enough $k$ and $n$ with probability one, the rate of $\mathcal{C}$ satisfies 
$$R\triangleq \frac{1}{n}\log_2 |\mathcal{C}|\approx \frac{k}{n}H(U|Q) .$$   
In what follows, we present a packing and a covering bound for the above code $\mathcal{C}$.


\begin{lem}[{Packing bound, \cite{QGC-ISIT16}}]\label{lem: packing}
 Let $(X,Y)$ distributed according to $p(x)p(y|x)$, for $x\in \ZZ_{p^r}$, and $y\in \mathcal{Y}$. By $\mathbf{\omega}_1$ denote the first codeword of  $ \mathcal{C}$. Let $\tilde{\mathbf{Y}}^n$ be a random sequence distributed according to $\prod_{i=1}^n p(\tilde{y}_i|\mathbf{\omega}_1)$. Suppose, conditioned on $\mathbf{\omega}_1$, the sequence $\tilde{\mathbf{Y}}^n$ is independent of other codewords in $\mathcal{C}$. Then, as $n\rightarrow \infty$,  $P\{\exists \mathbf{x}\in \mathcal{C}: (\mathbf{x}, \tilde{\mathbf{Y}}^n)\in A_{\epsilon}^{(n)} (X,Y), \mathbf{x}\neq \mathbf{\omega}_1\}$ is arbitrary close to zero, if
 \begin{align}\label{eq: packing bound}
R < \min_{0 \leq t\leq r-1} \frac{H(U|Q)}{H(U|Q,[U]_t)}\big( \log_2p^{r-t}-H(X|Y[X]_t) \big).
\end{align}
\end{lem}

\begin{lem}[{Covering bound, \cite{QGC-ISIT16}}]\label{lem: covering}
Suppose the pair of random variables $(X,\hat{X})$ are distributed according to $p(x,\hat{x})$,  where $\hat{X}$ takes values from $\ZZ_{p^r}$, and $X$ takes values from $\mathcal{X}$.  Let $\mathbf{X}^n$ be a random sequence distributed according to $\prod_{i=1}^n p(x_i)$. Then, as $n \rightarrow \infty$, $ P\{ \exists \hat{\mathbf{x}} \in \mathcal{C}: (\mathbf{X}^n,\mathbf{\hat{x}})\in A_{\epsilon}^{(n)} (X,\hat{X})\} $
is arbitrary close to one, if 
\begin{align}\label{eq: covering bound}
  R > \max_{1\leq t \leq r} \frac{H(U|Q)}{H([U]_t|Q)} (\log_2 p^t-H([\hat{X}]_t|X)).
\end{align}
\end{lem}

\section{Main Results}
We first propose a structured coding scheme that builds upon QGCs. Next, we present a method for \textit{binning} using QGCs. Then, we derive the single-letter characterization of the achievable rate region using such scheme.

  Consider a QGC defined by 
\begin{align}\label{eq: nested QGC codebook}
\mathcal{C}_O\triangleq \{\mathbf{u}\mathbf{G}+ \mathbf{v}\mathbf{\tilde{G}} +\mathbf{b}: \mathbf{u}\in \mathcal{U}, \mathbf{v}\in \mathcal{V}\},
\end{align}
where $\mathcal{U}$ and $\mathcal{V}$ are subsets of $\ZZ_{p^r}^k$, and $\ZZ_{p^r}^l$, respectively. Also $\mathbf{G}$ and $\mathbf{\tilde{G}}$ are $k\times n$ and $l\times n$ matrices, respectively. In this case, $\mathcal{C}_O$ is a $(n,k+l)$-QGC. We can associate an inner code for $\mathcal{C}_O$. Define the inner code as 
\begin{align*}
\mathcal{C}_I \triangleq \{\mathbf{u}\mathbf{G}+\mathbf{b}: \mathbf{u}\in \mathcal{U}\}.
\end{align*}
Therefore, $\mathcal{C}_I$ is a $(n,k)$-QGC, and $\mathcal{C}_I \subset \mathcal{C}_O$. The pair $(\mathcal{C}_I, \mathcal{C}_O)$ is called a nested QGC.    
\begin{definition}
A nested $(n, k, l)$-QGC is defined as
\begin{align}\label{eq: nested QGC codebook}
\mathcal{C}_O=\{\mathbf{x}_I\oplus \mathbf{\bar{x}}: \mathbf{x}_I \in \mathcal{C}_I, \mathbf{\bar{x}}\in \bar{\mathcal{C}}\},
\end{align}
where $\mathcal{C}_I$ is a $(n,k)$-QGC, and $ \bar{\mathcal{C}}$ is a $(n,l)$-QGC.  
\end{definition}

For any fixed element $\mathbf{u}\in \mathcal{U}$, we define its corresponding bin as the set 
\begin{align*}
\mathcal{B}(\mathbf{u})\triangleq \{\mathbf{u}\mathbf{G}+ \mathbf{v}\mathbf{\tilde{G}} +\mathbf{b}:  \mathbf{v}\in \mathcal{V}\}.
\end{align*} 
In this situation, $\mathcal{C}_O$ is binned using $\mathcal{C}_I$ as the inner code and $\mathcal{B}({\mathbf{u}})$ as the bins. Using this binning method, a rate region is given in the following Theorem. 
\begin{theorem}\label{thm: QGC MAC with state achievable}
For a given MAC $(\mathcal{X}_1,\mathcal{X}_2, \mathcal{Y}, P_{Y|X_1X_2})$ with independent states $(S_1,S_2)$ and cost functions $c_1,c_2$, the following rates are achievable using nested-QGCs 
\begin{align*}
R_1+R_2 &\leq r\log_2p - H(V_1\oplus V_2|YQ) - \max_{\substack{i=1,2\\ 1\leq t \leq r}}\Big\{\frac{H(W_1\oplus W_2|Q)}{H([W_i]_t|Q)} \Big(\log_2 p^t-H([V_i]_t| Q S_i)\Big)\Big\},
\end{align*}
where the joint distribution of the above random variables factors as $$p(q)p(s_1,s_2)\prod_{i=1,2} p(w_i|q) p(v_i|q,s_i)p(x_i|q, v_i, s_i) p(y|x_1,x_2).$$
\end{theorem}

\begin{proof}
Fix positive integers $n, k_1, k_2$, and $l$. Let $\mathcal{C}_{I,j}$ be a $(n,k_j)$-QGC with matrix $\mathbf{G}_j$, translation $\mathbf{b}_j$, and random variables $(Q_j,U_j)$, where $U_j$ is uniform over $\{0,1\}$, and $j=1,2$.  Let $\bar{\mathcal{C}}_1$ and $\bar{\mathcal{C}}_2$ be two $(n,l)$ QGC with identical matrices $\mathbf{\bar{G}}$ and identical translations $\mathbf{\bar{b}}$. Suppose $(\bar{Q}, W_j)$ are the random variables associated with $\bar{\mathcal{C}}_j$, where  $W_j$ takes values from $\ZZ_{p^r}$, and $j=1,2$. By $\mathcal{W}_1$ and $\mathcal{W}_2$ denote the sets corresponding to $\bar{\mathcal{C}}_1$ and $\bar{\mathcal{C}}_2$, respectively. Since $\bar{\mathcal{C}}_1$ and $\bar{\mathcal{C}}_2$ have identical matrices and translations, then $\bar{\mathcal{C}}_1\oplus \bar{\mathcal{C}}_2 $ is a $(n,l)$-QGC. The corresponding set of such sum-codebook is $\mathcal{W}_1\oplus \mathcal{W}_2$. Note that the elements of all the matrices and the translations are selected randomly and uniformly from $\ZZ_{p^r}$. 

\textbf{Codebook Construction:}
For each encoder we use a nested QGC. For the first encoder, we use the $(n,k_1,l)$-nested QGC generated by $\mathcal{C}_{I,1}$ and $\bar{\mathcal{C}}_1$. For the second encoder, we use the $(n,k_2,l)$-nested QGC characterized by $\mathcal{C}_{I,2}$ and $\bar{\mathcal{C}}_2$. For the decoder, as a codebook, we use a $(n,k_1+k_2+l)$-nested QGC. This codebook is denoted by $\mathcal{D}$. The inner code is a $(n,k_1+k_2)$-QGC defined by  $\mathcal{C}_{I,1}\oplus \mathcal{C}_{I,2}$.  The outer code is a $(n,k_1+k_2+l)$-QGC defined by $\bar{\mathcal{C}}_1\oplus \bar{\mathcal{C}}_2 \oplus \mathcal{C}_{I,1}\oplus \mathcal{C}_{I,2}$.  
For $i=1,2$ and for each sequence $\mathbf{s}_i$ and $\mathbf{v}_i \in \ZZ_{p^r}^n$, generate a sequence $\mathbf{x}_i$ according to $\prod_{j=1}^n p(x_{ij}|s_{ij}, v_{ij})$. Denote such sequence by $x_i(\mathbf{s}_i, \mathbf{v}_i)$.

\textbf{Encoding:}
 Without loss of generality, we assume that each message is selected randomly and uniformly from $\{0,1\}^k$. For $i=1,2$, the $i$th encoder is given a message $\mathbf{u}_i \in \{0, 1\}^k$, and a state sequence $\mathbf{s}_i$ with length $n$. The encoder first calculates the bin associated with $\mathbf{u}_i$. Next, it finds a  codeword $\mathbf{v}_i$ in the bin such that $(\mathbf{v}_i, \mathbf{s}_i)$ are jointly $\epsilon$-typical with respect to $P_{V_iS_i}$.  If no such sequence was found, the error event $E_i$ will be declared. If there was no error,  the $i$th encoder sends $x_i(\mathbf{s}_i, \mathbf{v}_i) ~ i=1,2$. The effective transmission rate for the $i$th encoder is $R_i=\frac{k_i}{n}, ~i=1,2$.
 
\textbf{Decoding:} 
We use $\mathcal{D}$ as a codebook in the receiver.  For each $\mathbf{\tilde{u}}_1, \mathbf{\tilde{u}}_2 \in \{0,1\}^k$ and $\mathbf{\tilde{w}}\in \mathcal{W}_1\oplus \mathcal{W}_2$  the decoder calculates the corresponding codeword defined as $$\mathbf{\tilde{v}}=\mathbf{\tilde{u}}_1\mathbf{G}_1+\mathbf{\tilde{u}}_2\mathbf{G}_2+\mathbf{\tilde{w}}\mathbf{\bar{G}}+\mathbf{b}_1+\mathbf{b}_2+\bar{\mathbf{b}}.$$ Upon receiving $\mathbf{Y}^n$ from the channel, it finds all $\mathbf{\tilde{v}}$ that are jointly $\epsilon$-typical with $ \mathbf{Y}^n$ with respect to $P_{V_1\oplus V_2, Y}$.  If the corresponding  $(\mathbf{\tilde{u}}_1, \mathbf{\tilde{u}}_2)$ sequences are unique,  they will be declared as the decoded messages. Otherwise, an error event $E_d$ will be announced.
 
 \textbf{Error Analysis:}
Let $\rho_1$ and $\rho_2$ denote the rate of  $\bar{\mathcal{C}}_1$ and $\bar{\mathcal{C}}_2$, respectively. We use Lemma \ref{lem: covering} to analyze the probability of $E_1$ and $E_2$. In this lemma, set $\mathcal{C}=\bar{\mathcal{C}}_1,  \hat{X}=V_i$, and $X=S_i$. Note that in this case, $E_i$ is the same as the event described in the Lemma. As a result,  we use the covering bound in (\ref{eq: covering bound}), where $R=\rho_i, U=W_i, Q=\bar{Q}, \hat{X}=V_i$, and $X=S_i, ~i=1,2$. Therefore, according to Lemma \ref{lem: covering}, $P(E_i)$ approaches zero as $n\rightarrow \infty$, if the following bound holds:
\begin{align}\label{eq: covering bound for encoder i}
  \rho_i > \max_{1\leq t \leq r} \frac{H(W_i|\bar{Q})}{H([W_i]_t|\bar{Q})} \big(\log_2 p^t-H([V_i]_t|S_i)\big).
\end{align}
Next, we use Lemma \ref{lem: packing} to bound the probability of the event $E_d$. In this lemma set $\mathcal{C}=\mathcal{D}$, and $X=V_1\oplus V_2$. In this case, $E_d$ is the event defined in the Lemma. If $\rho$ is the rate of $\bar{\mathcal{C}}_1\oplus \bar{\mathcal{C}}_2$, then the rate of $\mathcal{D}$ equals $R_1+R_2+\rho$. As a result of Lemma \ref{lem: packing},  $P(E_d| E_1^c\cap E_2^c)$ approaches zero, if the packing bound in (\ref{eq: packing bound}) holds for $R=R_1+R_2+\rho, U=(U_1,U_2), Q=(Q_1,Q_2)$ Since $U_i$ is uniform over $\{0,1\}$, then $H(U_i|Q_i, [U_i]_t)=0$ for all $t>0$. Therefore, the packing bound is simplified to 
\begin{align}\label{eq: packing for the decoder }
R_1+R_2+\rho \leq \log_2p^r-H(V_1\oplus V_2|Y).
\end{align}
It can be shown that $\rho=\frac{H(V_1\oplus V_2|\bar{Q})}{H(V_i|\bar{Q})}\rho_i$. Finally the bound in the theorem follows by using this equality,  bounds in (\ref{eq: covering bound for encoder i}) and (\ref{eq: packing for the decoder }), and denoting $Q=(Q_1,Q_2,\bar{Q})$.
\end{proof}

\begin{corollary}
Set $V_i\sim unif(\ZZ_{p^r}), i=1,2$. Then the rate-region in the Theorem is simplified to the achievable rate region of group codes, that is 
\begin{align*}
R_1+R_2 &\leq \min_{\substack{i=1,2\\ 1\leq t \leq r}}\{H([V_i]_t| Q S_i)\} - H(V_1\oplus V_2|YQ) .
\end{align*}
\end{corollary}

We proposed a coding strategy using nested QGCs to achieve the rate region presented in Theorem \ref{thm: QGC MAC with state achievable}. We build upon this coding scheme and the extension of the Gel'fand-Pinsker scheme, and propose a new coding strategy. Using this scheme, a new achievable rate region is characterized in the next Theorem.   

\begin{theorem}\label{thm: QGC+unstructured ahcievable}
For a given MAC $(\mathcal{X}_1,\mathcal{X}_2, \mathcal{Y}, P_{Y|X_1X_2})$ with independent states $(S_1,S_2)$ and cost functions $c_1,c_2$,  the following rate region is achievable 
\begin{align*}
R_1 &\leq I(U_1;Y|U_2Q)-I(U_1;S_1|Q) +\Gamma_{QGC}\\
R_1 &\leq I(U_2;Y|U_1Q)-I(U_2;S_2|Q) +\Gamma_{QGC}\\
R_1+R_2 &\leq I(U_1U_2;Y|Q)-I(U_1U_2;S_1S_2|Q)+\Gamma_{QGC},
\end{align*}
where 
\begin{align*}
\Gamma_{QGC}& \triangleq r\log_2p - H(V_1\oplus V_2|YU_1U_2Q) - \max_{\substack{i=1,2\\ 1\leq t \leq r}}\Big\{ \frac{H(W_1\oplus W_2|Q)}{H([W_i]_t|Q)} \Big(\log_2 p^t-H([V_i]_t|U_i Q S_i)\Big)\Big\},
\end{align*}
and 1) the cost constraints $\EE\{c_i(X_i,S_i)\}\leq \tau_i$ are satisfied, 2) the Markov chain $$(S_1,U_1,V_1,W_1,X_1)\leftrightarrow Q \leftrightarrow (S_2,U_2,V_2, W_2,X_2)$$ holds, 3)given $Q, X_1,X_2$ the random variable $Y$ is independent of all other random variables, and 3) conditioned on $Q$, the random variables $W_1,W_2$ are independent of other random variables.
\end{theorem}

\begin{proof}
The proof is provided in Appendix \ref{sec: proof of thm QGC+unstructured achievable}.
\end{proof}
\begin{remark}
The rate region presented in Theorem \ref{thm: QGC+unstructured ahcievable} contains the rate region presented in Proposition \ref{prep: Gelfand-pinker achievable}. 
\end{remark}

\section{An Example}
We present a MAC with state setup for which the Gel'fand-Pinsker region given in Proposition \ref{prep: Gelfand-pinker achievable} is strictly contained the region given in Theorem \ref{thm: QGC+unstructured ahcievable}.
\begin{example}\label{ex: MAC states}
Consider a noiseless MAC described by $$Y=X_1\oplus_4 S_1 \oplus_4 X_2 \oplus_4 S_2,$$
where $X_1, X_2$ are the inputs, $Y$ is the output, and $S_1,S_2$ are the states. All the random variables take values from $\ZZ_4$. The states $S_1$ and $S_2$ are mutually independent, and are distributed uniformly over $\ZZ_4$. The addition $\oplus_4$ is the modulo-$4$ addition. The cost function at the first encoder is defined as 
\begin{align*}
c_1(x)\triangleq \left\{
                \begin{array}{ll}
                  1 & \text{if}~ x\in \{1,3\} \\
                  0 & \text{otherwise,}
                  \end{array}
              \right.
\end{align*}
whereas, for the second encoder the cost function is 
\begin{align*}
c_2(x)\triangleq \left\{
                \begin{array}{ll}
                  1 & \text{if}~ x\in \{2,3\} \\
                  0 & \text{otherwise.}
                  \end{array}
              \right.
\end{align*}
We are interested in satisfying the  cost constraints $\EE\{c_1(X_1)\}=\EE\{c_2(X_2)\}=0$. This implies that, with probability one, $X_1\in \{0, 2\}$, and  $X_2\in \{0, 1\}$.
\end{example}

We proceed using two lemmas. First, we derive an outer-bound on the Gel'fand-Pincker region. Then, we show that the outer-bound is strictly contained in the achievable rate region using QGC.

\begin{lem}\label{lem: suboptinality of Gelfand-Pinsker}
For the setup in Example \ref{ex: MAC states},  an outer-bound on the Gel'fand-Pinsker region given in Proposition \ref{prep: Gelfand-pinker achievable} is the set of all rate pairs $(R_1, R_2)$ such that $R_1+R_2\leq 0.32$.
\end{lem}
\begin{proof}
The proof is given in the Appendix \ref{sec: proof of lem suboptinality of Gelfand-Pinsker }.
\end{proof}
\begin{lem}
For the setup in Example 1, the rate pairs $(R_1,R_2)$ satisfying $R_1+R_2= 1$ is achievable using QGCs. 
\end{lem}
\begin{proof}
We use the proposed scheme presented in the proof of Theorem \ref{thm: QGC MAC with state achievable}. Similar to the proof of the Theorem, two $(n, k, l)$ nested QGCs are used, one for each encoder.  Set $W_1$ and $W_2$, the random variables associated with the QGC, to be distributed uniformly over $\{0,1\}$. Suppose $\mathbf{v}_1, \mathbf{v}_2$ are the output of the nested-QGC at encoder 1 and encoder 2, respectively. Encoder 1 sends $\mathbf{x}_1=\mathbf{v}_1\ominus \mathbf{s}_1$, where $\mathbf{s}_1$ is the realization of the state $S_1$. Similarly, the second encoder sends $\mathbf{x}_2=\mathbf{v}_2\ominus \mathbf{s}_2$, where $\mathbf{s}_2$ is the realization of the state $S_2$. The conditional distribution of $v_1$ given $s_1$ is 
\begin{align*}
p(v_1|s_1)\triangleq \left\{
                \begin{array}{ll}
                  1/2 & \text{if} ~ v_1=-s_1, \text{or} ~ v_1=-s_1\oplus 2 \\
                  0 & \text{otherwise},
                  \end{array}
              \right.
\end{align*}
The distribution of $V_2$ conditioned of $S_2$ is
\begin{align*}
p(v_2|s_2)\triangleq \left\{
                \begin{array}{ll}
                  1/2 & \text{if} ~ v_1=-s_1, \text{or} ~ v_1=-s_1\oplus 1 \\
                  0 & \text{otherwise},
                  \end{array}
              \right.
\end{align*}
As a result, $X_1\in \{0,2\}, X_2\in \{0,1\}$. Hence, the cost constraints are satisfied. In this situation,  $H([V_i]_1)=H(V_i)=1,$ for $i=1,2$, and $H(V_1\oplus V_2)=\frac{3}{2}$.  Therefore, assuming $Q$ is trivial, the sum-rate given in the Theorem is simplified to 
\begin{align*}
R_1+R_2 &\leq   \frac{3}{2}\min\{H(V_1|S_1), H(V_2|S_2)\} - H(V_1\oplus V_2|Y)-\frac{1}{2} =1,
\end{align*}
where the last equality holds, because $H(V_i|S_i)=1$, and $H(V_1\oplus V_2|Y)=H(X_1\oplus S_1\oplus  X_2\oplus S_2|Y)=0$. As a result the sum -rate $R_1+R_2=1$ is achievable.
\end{proof}

\section{Conclusion}
The problem of  non-binary MAC with states was investigated. We built upon QGC, and the extension of Gel'fand-Pinsker scheme, and propose a new coding scheme. Then, the single-letter characterization of the achievable region using this scheme was derived. We used the coding scheme for the doubly-dirty MAC. We proved that the proposed coding scheme strictly outperforms the  Gel'fand-Pinsker scheme.

\appendices
\section{Proof of Lemma \ref{lem: suboptinality of Gelfand-Pinsker}}\label{sec: proof of lem suboptinality of Gelfand-Pinsker }
\begin{proof}
In what follows, we give an upper-bound on (\ref{eq: sum-rate}). The time-sharing random variable $Q$ in Proposition \ref{prep: Gelfand-pinker achievable} is trivial, because of the cost constraints $\EE\{c_i(X_i)\}=0, i=1,2$.  For the bound (\ref{eq: sum-rate}), we obtain 
\begin{align*}
&R_1+R_2\leq  I(U_1 U_2; Y)-I(U_1;S_1)-I(U_2;S_2)\\
&\leq H(S_1|U_1)+H(S_2|U_2)-H(Y|U_1U_2)-2\\
&= \sum_{u_1, u_2} p(u_1,u_2) \Big( H(S_1|u_1)+H(S_2|u_2)-H(Y|u_1u_2)-2\Big)\\
&\leq \max_{u_1\in \mathcal{U}_1, u_2\in \mathcal{U}_2}\Big( H(S_1|u_1)+H(S_2|u_2)-H(Y|u_1u_2)-2\Big),
\end{align*}
where the second inequality holds, as $H(Y)\leq 2$, and $H(S_i)=2$ for $i=1,2$.  Let $\mathscr{P}$ be the collection of all valid PMFs used in Proposition \ref{prep: Gelfand-pinker achievable}. For any distribution $P\in \mathscr{P}$ define  
\begin{align*}
R(u_1,u_2, P) \triangleq H(S_1|u_1)+H(S_2|u_2)-H(Y|u_1u_2)-2
\end{align*} 

 In the next step, we relax the conditions in $\mathscr{P}$.  For $i=1,2$, and any $u_i\in \mathcal{U}_i$, define   $\mathscr{P}_{u_i}$ as the collection of all conditional pmfs $ p(s_i,x_i|u_i)$ on $\ZZ^2_4$ such that $E(c_i(X_i)|u_i)=0.$
This condition is obtained from the cost constraint $E(c_i(X_i))=0$ 
(because, without loss of generality we assume $p(u_i)>0, \forall u_i\in \mathcal{U}_i$). 
For any PMF  $P \in \mathscr{P}$, the states $S_1,S_2$ are independent, and the Markov chain $U_1X_1-S_1-S_2-U_2X_2$ holds. Therefore, $P$ factors as $\prod_{i=1}^2 p(u_i) p(s_i, x_i|u_i)$, where $p(s_i, x_i|u_i)$ satisfies the conditions in the definition of $\mathscr{P}_{u_i}$. Hence, $\mathscr{P}$ is a subset of the set of all PMFs $\prod_{i=1}^2 p(u_i)p(s_i,x_i|u_i)$, where $p(s_i,x_i|u_i)\in \mathcal{P}_{u_i}$. 
%
%
As a result, we get 
\begin{align*}
&R_1+R_2\\
&\leq  \max_{p(u_1), p(u_2)}\max_{\substack{ p(s_i,x_i|u_i)\in\mathcal{P}_{u_i}\\ i=1,2 }} \sum_{u_1,u_2} p(u_1,u_2) R(u_1,u_2,P)\\
&\leq \sum_{u_1,u_2}  \max_{p(u_1), p(u_2)}\max_{\substack{ p(s_i,x_i|u_i)\in\mathcal{P}_{u_i}\\ i=1,2 }}p(u_1,u_2) R(u_1,u_2,P)\\
&\leq \sum_{u_1,u_2}  \max_{p(u_1), p(u_2)}p(u_1,u_2)\max_{\substack{ p(s_i,x_i|u_i)\in\mathcal{P}_{u_i}\\ i=1,2 }}R(u_1,u_2,P)\\
&\leq\max_{u_1 \in \mathcal{U}_1, u_2\in \mathcal{U}_2} \max_{\substack{ p(s_i,x_i|u_i)\in\mathcal{P}_{u_i}\\ i=1,2 }}R(u_1,u_2,P)
\end{align*}
Fix $u_2\in \mathcal{U}_2$ and $p(s_2,x_2|u_2)\in\mathcal{P}_{u_2}$. We maximize over all $u_1\in \mathcal{U}_1$ and $p(s_1,x_1|u_1)\in\mathcal{P}_{u_1}$. By $Q_{u_2}\in \mathcal{P}_{u_2}$ denote the PMF $p(s_2,x_2|u_2)$. This optimization problem is equivalent to the following problem
\begin{align*}
R(u_2, Q_{u_2})= H(S_2|u_2)+\max_{u_1\in \mathcal{U}_1} \max_{Q \in \mathcal{P}_{u_1} }H(S_1|u_1)-H(Y |u_1)-2.
 \end{align*}
Let $N=X_2\oplus S_2$, where $X_2$ and $S_2$ are distributed according to  $p(s_2,x_2|u_2)$. Consider the problem of ptp channel with state, where the channel is $Y=X_1\oplus S_1 \oplus N$.  It can be shown that the above quantity is an upper-bound on the capacity of this problem.
The following lemma completes the proof.
\begin{lem}\label{lem: R(u_2, Q)< 0.32}
\small{$R(u_2, Q_{u_2})\leq 0.32$} for all $u_2\in \mathcal{U}_2$ and $Q_{u_2}\in \mathcal{P}_{u_2}.$  
\end{lem}
%
The proof of this lemma is given in Appendix \ref{sec: proof of lem R(u_2, Q)< 0.32}.
\end{proof}

\section{Proof of Theorem \ref{thm: QGC+unstructured ahcievable}} \label{sec: proof of thm QGC+unstructured achievable}
\begin{proof}
We propose a coding scheme which is a combination of two coding schemes: 1) Gel'fand-Pinsker scheme, and 2) the proposed scheme in Theorem \ref{thm: QGC MAC with state achievable} which is uses nested QGCs. Suppose $M_j$ is the message for the $j$th user. $M_j$ is drawn randomly and uniformly from $[1:2^{R_j}]$. The $j$th encoder splits its message $M_j$ into two parts $M_{j,1}$ and $M_{j,2}$, where $j=1,2$. Suppose $M_{j,1}\in [1: 2^{nR_{j,1}}] $ and $M_{j,2}\in [1: 2^{nR_{j,2}}]$, where $R_j=R_{j,1}+R_{j,2}$. The first part $M_{j,1}$ is encoded  using the natural extension of Gel'fand-Pinsker. The second part $M_{j,2}$ is encoded using a nested QGC as described in the proof of Theorem \ref{thm: QGC MAC with state achievable}. 

\textbf{Codebook Construction:}
\begin{itemize}
\item   For each $j=1,2$ and any $m_{j,1}$  generate $2^{\rho_{j,1}}$ sequences $\mathbf{u}_{j}^n$ randomly and independently according to the distribution $\prod_{i=1}^n p(u_{j,i})$. Such sequences are denoted by $u_j(m_{j,1}, a_j)$, where $a_j\in [1:2^{n\rho_{j1}}]$. The collection of all such codewords is denote by $\mathcal{C}_{j,1}$. 

\item We use a $(n,k_j,l)$-nested QGC as described in the proof of Theorem  \ref{thm: QGC MAC with state achievable}. Denote such nested QGC by $\mathcal{C}_{j,2}$. Let $\mathcal{C}_{I,j}$ be the inner codebook associated to $\mathcal{C}_{j,2}$. Let $2^{n\rho_{j,2}}$ be the size of $\mathcal{C}_{I,j}$. As described in the proof of Theorem \ref{thm: QGC MAC with state achievable}, the codebook $\mathcal{C}_{j,2}$ is divided into $2^{nR_{j,2}}$ bins, where each bin is a shifted version of the inner codebook. Each bin corresponds to a message $m_{j,2}\in [1:2^{nR_{j,2}}]$. Denote such bin by $\mathcal{B}_j(m_{j,2})$.

\item Given the sequences $\mathbf{s}_j \in \mathcal{S}_j^n, \mathbf{u}_j\in \mathcal{U}_j^n$, and $\mathbf{v}_j\in \ZZ_{p^r}^n$ generate a sequence $\mathbf{x}_j$ according to $\prod_{j=1}^n p(x_{j,i}|s_{j,i} u_{j,i}, v_{j,i})$. Denote such sequence by $x_j(\mathbf{s}_j, \mathbf{u}_j, \mathbf{v}_j)$. 

\item For the decoder, we use $\mathcal{C}_{1,1}, \mathcal{C}_{2,1}$ and $\mathcal{D}$ as the codebooks, where  $\mathcal{D}=\mathcal{C}_{1,2}\oplus \mathcal{C}_{2,2}$. Note $\mathcal{D}$ is a $(n, k_1+k_2, l)$ nested QGC. The inner code associated with $\mathcal{D}$ is $\mathcal{C}_{I,1}\oplus \mathcal{C}_{I,2}$. Let $2^{n\rho}$ denote the size of the inner code. There are $2^{n(R_{1,2}+R_{2,2})}$ bins in $\mathcal{D}$. Each bin corresponds to a message pair $(m_{1,2}, m_{2,2})$. 

\end{itemize}
 
\textbf{Encoding:}
The $j$th encoder is given a message  pair $(m_{j,1}, m_{j,2})$ and a state sequence $\mathbf{s}_j$. The $j$th encoder finds  $\mathbf{v}_j\in \mathcal{B}_j(m_{j,2})$ and $a_j\in [1:2^{n\rho_{j1}}]$ such that $( u_j(m_{j,1}, a_j),\mathbf{v}_j, \mathbf{s}_j)\in A_\epsilon^{(n)}(U_j,V_j,S_j) $. If such sequences were found, the $j$th encoder sends $x_j(\mathbf{s}_j, \mathbf{u}_j, \mathbf{v}_j)$, where $\mathbf{u}_j=u_j(m_{j,1}, a_j)$. Otherwise an error is declared.

\textbf{Decoding}:
The decoder receives $\mathbf{Y}^n$ from the channel. The decoding is performed in two stages. In the first stage, the decoder lists all codewords $\mathbf{\tilde{u}}_1\in \mathcal{C}_{1,1}, \mathbf{\tilde{u}}_2\in \mathcal{C}_{2,1}$ such that    $(\mathbf{\tilde{u}}_1,\mathbf{\tilde{u}}_2, Y^n)$ are $\epsilon$- typical with respect to $P_{U_1U_2Y}$. If $\mathbf{\tilde{u}}_1, \mathbf{\tilde{u}}_2$ are unique, the decoder proceeds to the next stage. Otherwise it declares an error. At the next stage, the decoder finds all $\mathbf{\tilde{v}}\in \mathcal{D}$ such that $(\mathbf{\tilde{u}}_1,\mathbf{\tilde{u}}_2, \mathbf{\tilde{v}}, Y^n) \in A_\epsilon^{(n)}( U_1U_2V_1\oplus V_2Y)$. Then the decoder checks if all $\mathbf{\tilde{v}}$ belong to a unique bin associated with $(\tilde{m}_{1,2}, \tilde{m}_{2,2})$. Finally the decoder declares that $(\tilde{m}_{1,1}, \tilde{m}_{1,2}, \tilde{m}_{2,1}, \tilde{m}_{2,2})$ is sent, if it is unique. Otherwise it declares an error.

\textbf{Error Analysis:}
We can show that the probability of error at the encoders is small enough, if the following covering bounds hold
\begin{align*}
\rho_{j,1} &> I(U_j;S_j)\\
\rho_{j,2}& >  \max_{1\leq t \leq r} \frac{H(W_j|\bar{Q})}{H([W_j]_t|\bar{Q})} \big(\log_2 p^t-H([V_j]_t|S_j U_j)\big),
\end{align*}
where $j=1,2$. Also the error at the decoder is small, if the following packing bounds hold
\begin{align*}
R_{1,1}+\rho_{11}&< I(U_1;Y|U_2)\\
R_{2,1}+\rho_{21}&< I(U_2;Y|U_1)\\
R_{1,1}+\rho_{11}+R_{2,1}+\rho_{21}&<I(U_1U_2;Y)+I(U_1;U_2)\\
R_{1,2}+R_{2,2}+\rho &< \log_2p^r-H(V_1\oplus V_2|YU_1U_2),
\end{align*}
where $\rho=\frac{H(V_1\oplus V_2|\bar{Q})}{H(V_j|\bar{Q})}\rho_{j,2}, j=1,2$. Next, we substitute $R_j-R_{j,2}$ for $R_{j,1}, j=1,2$ in the above bounds. Finally, we use the Fourier-Motzkin technique \cite{ElGamal-book} to eliminate $R_{j,2}, \rho_{j,1}, \rho_{j,2}, j=1,2$. This completes the proof. 
\end{proof}
\section{Proof of Lemma \ref{lem: R(u_2, Q)< 0.32}}\label{sec: proof of lem R(u_2, Q)< 0.32}
\begin{proof}
Note that for any fixed $u_2\in \mathcal{U}_2$, the distribution of $N$ depends on  the conditional pmf $p(s_1|u_1)$, and the function $x_1=f(s_1,u_1)$.  For any $u\in \mathcal{U}_2$ define  
$$\mathcal{L}_u:=\{f_2(u,s)\oplus s: s\in \ZZ_4\}.$$
For any given $i\in \{1,2,3,4\}$,  define
$$\mathcal{B}_i \triangleq \{u\in \mathcal{U}_2:    |\mathcal{L}_u|=i   \}.$$
Note that $\mathcal{B}_i$'s are disjoint and $\mathcal{U}_2=\bigcup_i \mathcal{B}_i$. Depending on $u_2$, we consider four cases. In what follows, for each case, we derive an upper bound on  $R(u_2, Q_{u_2})$. Consider the pmf $p(\omega)$ on $\ZZ_4$. For brevity, we represent this pmf by the vector $\mathbf{p}:=(p(0), p(1), p(2), p(3))$. 

\subsection*{ Case 1: $u_2\in \mathcal{B}_1$}
Since $|\mathcal{L}_{u_2}|=1$, then for all $s_2\in \ZZ_4$ the equality $s_2\oplus f_2(s_2,u_2)=a$ holds, where $a\in\ZZ_4$ is a constant that only depends on $u_2$. This implies that conditioned on $u_2$, $X_2\oplus S_2$ equals to a constant $a$, with probability one. Therefore,
\begin{align*}
H(X_1\oplus S_1 \oplus X_2 \oplus S_2|u_2 u_1)=H(X_1\oplus S_1\oplus a|u_1u_2)=H(X_1\oplus S_1|u_1)
\end{align*}
Moreover, $$H(S_2|u_2)=H(a \ominus X_2|u_2)=H(X_2|u_2)\leq H(X_2)\leq 1,$$ where the last inequality holds, because of the cost constraint $\EE(w_2(X_2))=0$. As a result, 
\begin{align*}
R(u_2, Q_{u_2}) \leq H(S_1|u_1)-H(X_1\oplus S_1|u_1)-1
\end{align*}
We show in Lemma \ref{lem: inequalities for ptp} that the right-hand side equals $0$.

\subsection*{ Case 2: $u_2\in \mathcal{B}_2$}
For any fixed $u_2\in \mathcal{B}_2$, $f_2(s_2, u_2) \oplus s_2$ takes two values for all $s_2\in \ZZ_4$. Assume these values are $a, b\in \ZZ_4$, where $a\neq b$. Given $u_2$ the random variable  $X_2\oplus S_2$ is distributed over $\{a,b\}$. Therefore, $X_2 \oplus S_2\ominus a$ is distributed over $\{0, b\ominus a\}$, and  
\begin{align*}
H(X_1\oplus S_1 \oplus X_2 \oplus S_2|u_2 u_1)=H(X_1\oplus S_1 \oplus X_2 \oplus S_2\ominus a |u_2 u_1).
\end{align*}
As a result, the case $\{a,b\}$ gives the same bound as $\{0, b\ominus a\}$, and we need to consider only the case in which $a=0$. For the case in which $a=0$, and  $b=3$, consider $X_2 \oplus S_2\oplus 1$. Using a similar argument as above, we can show that when $b=3$, we get the same bound when $b=1$. Therefore, we only need to consider the cases in which $a=0$, and $b\in \{1, 2\}$. We address these cases in the next Lemma.
%
\begin{lem}\label{lem: Case 2}

Let $P(X_2 \oplus S_2 =0 |u_1)=p_0$. The following holds: 

1) If $b=2$, then
\begin{align*}
R(u_2, Q_{u_2}) &\leq \beta(H(S_1|u_1)-H(X_1\oplus S_1 \oplus N_{(2/3,0,1/3,0)}|u_1))\\
&+(1-\beta)(H(S_1|u_1)-H(X_1\oplus S_1\oplus N_{(1/3,0,2/3,0)}|u_1))+H(S_2|u_2)-2
\end{align*}

2) If $b=1$, then 
\begin{align*}
R(u_2, Q_{u_2})& \leq  \beta(H(S_1|u_1)-H(X_1\oplus S_1 \oplus N_{(2/3,1/3,0,0)}|u_1))\\
&+(1-\beta)(H(S_1|u_1)-H(X_1\oplus S_1\oplus N_{(1/3,2/3,0,0)}|u_1))+H(S_2|u_2)-2
\end{align*}
\end{lem}
\begin{proof}
The proof is given in Appendix \ref{sec: proof of lemma case 2}.
\end{proof}

Using  Lemma  \ref{lem: inequalities for ptp}, we show that  $$H(S_1|u_1)-H(X_1\oplus S_1 \oplus N_{(2/3,0,1/3,0)}|u_1)\leq 0.1,$$ and 
$H(S_1|u_1)-H(X_1\oplus S_1\oplus N_{(1/3,0,2/3,0)}|u_1) \leq 0.1$. Therefore, if $a=0, b=2$, we have
\begin{align*}
R(u_2, Q_{u_2})&\leq  0.1 + H(S_2|u_2)-2\leq  0.1,
\end{align*}
where the last inequality holds, because $H(S_2|u_2)\leq H(S_2)=2$.

For the case in which $a=0, b=1$, from numerical calculations in Lemma \ref{lem: inequalities for ptp},  we can show that  $$H(S_1|u_1)-H(X_1\oplus S_1 \oplus N_{(2/3,1/3,0,0)}|u_1)\leq 0.5,$$ and $$H(S_1|u_1)-H(X_1\oplus S_1\oplus N_{(1/3,2/3,0,0)}|u_1)\leq 0.5.$$ Therefore, 
\begin{align*}
R(u_2, Q_{u_2}) & \leq   H(S_2|u_2)-1.5
\end{align*}
By an extensive search over all functions in this case that satisfy the cons constrains, we can show that given $u_2$ the random variable $S_2$ can take at most 3 values with positive probabilities. Thus, in this situation  $H(S_2|u_2)\leq \log_2 3$, and 
\begin{align*}
R(u_2, Q_{u_2}) &\leq  \log_2 3 - 1.5 \approx 0.09.
\end{align*}

 \subsection*{ Case 3: $u_2\in \mathcal{B}_3$}
We need only to consider the case when $\mathbf{p}=(p_0, p_1, p_2, 0)$. We have

\begin{lem}
If $u_2\in \mathcal{B}_3$, the following bound holds
\begin{align*}
R(u_2, Q_{u_2}) & \leq \beta_0 (H(S_1|u_1)-H(X_1\oplus S_1 \oplus N_{(2/4,1/4,1/4, 0)}|u_1))\\
&+\beta_1(H(S_1|u_1)-H(X_1\oplus S_1\oplus N_{ (1/4,2/4,1/4,0)}|u_1))\\
&+\beta_2 (H(S_1|u_1)-H(X_1\oplus S_1\oplus N_{ (1/4,1/4,2/4,0)}|u_1))+H(S_2|u_2)-2,
\end{align*}
 where $\beta_i=4p_i-1, ~i=0,1,2$.
\end{lem}

\begin{proof}
Similar to Case 2, we can write $\mathbf{p}$ as a linear combination of three distributions of the form $$\mathbf{p}=\beta_0 (2/4,1/4,1/4, 0)+\beta_1 (1/4,2/4,1/4,0)+\beta_2 (1/4,1/4,2/4,0),$$ where $\beta_i=4p_i-1, ~i=0,1,2$. The proof then follows from the concavity of the entropy.
\end{proof} 
Using Lemma \ref{lem: inequalities for ptp}, we obtain 
 \begin{align*}
R(u_2, Q_{u_2}) &\leq   0.32 +H(S_2|u_2)-2 \leq 0.32
\end{align*}
 
 \subsection*{Case 4: $u_2\in \mathcal{B}_4$}
 In this case, there is a 1-1 correspondence between $x_2(s_2,u_2)\oplus s_2$ and $s_2$. Therefore $H(S_2|u_2)=H(S_2\oplus X_2|u_2)$, and we obtain
\begin{align*}
H(S_2|u_2)-H(X_1\oplus S_1 \oplus X_2 \oplus S_2|u_1)&=H(S_2\oplus X_2|u_2)-H(X_1\oplus S_1 \oplus X_2 \oplus S_2|u_1)\\&\leq 0
\end{align*} 
Therefore 
$H(S_1|u_1)+H(S_2|u_2)-H(Y|u_1u_2)-2\leq  H(S_1|u_1)-2\leq 0.$

Finally, considering all four cases $R(u_2, Q_{u_2}) \leq 0.32$ for all $u_1 \in \mathcal{U}_1$ and $u_2 \in \mathcal{U}_2$. This completes the proof.  
\end{proof}
\section{Lemma \ref{lem: inequalities for ptp}}
\begin{lem}\label{lem: inequalities for ptp}
Suppose $p(\omega)$ is a PMF on $\ZZ_4$. By $N_{\mathbf{p}}$ denote a random variable with distribution $p$ that is independent of $S$. Then for any function $x(s)$, and any PMF $p(s)$ satisfying $\EE\{w_1(X)\}=0$, the following bounds hold:
\begin{align*}
H(S)-H(X\oplus S)&\leq 1\\
H(S)-H(X\oplus S \oplus N_{(1/3, 0, 2/3,0)}|u_1)&\leq 0.1\\
H(S)-H(X\oplus S \oplus N_{(2/3, 0, 1/3,0)}|u_1)&\leq 0.1\\
H(S)-H(X\oplus S \oplus N_{(1/3, 2/3,0,0)}|u_1)&\leq 0.5 \\
H(S)-H(X\oplus S \oplus N_{(2/3, 1/3,0,0)}|u_1)&\leq 0.5\\
H(S)-H(X\oplus S \oplus N_{ (2/4,1/4,1/4,0)}|u_1)&\leq 0.32\\
H(S)-H(X\oplus S \oplus N_{ (1/4,2/4,1/4,0)}|u_1)&\leq 0.32\\
H(S)-H(X\oplus S \oplus N_{ (1/4,1/4,2/4,0)}|u_1)&\leq 0.32
\end{align*}
\end{lem}
\begin{proof}
The proof follows by numerically calculating the left-hand side of any bound at any PMF $p$ and any function $x(s)$ .  
\end{proof}

\section{Proof of Lemma \ref{lem: Case 2}}\label{sec: proof of lemma case 2}
\begin{proof}
\paragraph*{1)}
 Let $a=0, b=2$, and  $P(X_2 \oplus S_2 =0 |u_1)=p_0$, and $P(X_2 \oplus S_2 =2 |u_1)=1-p_0$. We represent this pmf by the vector $\mathbf{p}=(p_0, 0, 1-p_0, 0)$. This probability distribution is a linear combination of the form 
\begin{align}\label{eq: p_case1}
\mathbf{p}=\beta (2/3, 0, 1/3, 0)+(1-\beta)(1/3,0,2/3,0),
\end{align} 
where $\beta=3p_0-1$.
 \begin{remark} \label{rem: circular conv and bilinear}
Let $Z=X\oplus Y$, where the pmf of $X$ is $\mathbf{p}=(p_0, p_1, p_2,p_3)$, and the pmf of $Y$ is $\mathbf{q}=(q_0, q_1, q_2,q_3)$. If $\mathbf{t}$ is the pmf of $Z$, then $\mathbf{t}=\mathbf{p} \ocoasterisk_4 \mathbf{q} $, where $ \ocoasterisk_4$  is the circular convolution in $\ZZ_4$. In addition, the map $(\mathbf{p} , \mathbf{q})\longmapsto \mathbf{p} \ocoasterisk_4 \mathbf{q} $ is a bi-linear map.
\end{remark}
Let $t_i=p(X_1\oplus S_1 \oplus X_2\oplus S_2=i|u_1u_2)$ and $q_i=p(X_1\oplus S_1 =i |u_1)$ for all $i\in \ZZ_4$. Also denote $\mathbf{q}=(q_0, q_1, q_2,q_3)$, and $\mathbf{t}=(t_0, t_1, t_2, t_3)$. Using  Remark \ref{rem: circular conv and bilinear} and equation (\ref{eq: p_case1}) we obtain 
\begin{align*}
\mathbf{t}&=\beta \big((2/3, 0, 1/3, 0)\ocoasterisk_4 \mathbf{q}\big)+(1-\beta)\big((1/3, 0, 2/3, 0)\ocoasterisk_4 \mathbf{q}\big).
\end{align*}
%
This implies that, $\mathbf{t}$ is also a linear combination of two pmfs. From the concavity of entropy, we get the following lower-bound:
\begin{align*}
H(X_1\oplus S_1 &\oplus X_2\oplus S_2|u_1u_2)=H(\mathbf{t})\\
&=H(\beta \big((2/3, 0, 1/3, 0)\ocoasterisk_4 \mathbf{q}\big)+(1-\beta)\big((1/3, 0, 2/3, 0)\ocoasterisk_4 \mathbf{q}\big))\\
& \geq \beta H((2/3, 0, 1/3, 0)\ocoasterisk_4 \mathbf{q}) +(1-\beta)H((1/3, 0, 2/3, 0)\ocoasterisk_4 \mathbf{q})\\  
& = \beta H(X_1\oplus S_1 \oplus N_{(2/3, 0, 1/3, 0)}| u_1) +(1-\beta)H(X_1\oplus S_1\oplus N_{(1/3, 0, 2/3, 0)} | u_1),  
\end{align*}
where in the last equality, $N_{(\lambda_0, \lambda_1,\lambda_2, \lambda_3)}$ denotes a random variable with pmf $(\lambda_0, \lambda_1,\lambda_2, \lambda_3)$ that is also independent of $u_1$ and $X_1\oplus S_1$. As a result of the above argument, $R(u_2, Q_{u_2}) $ is bounded by
\begin{align*}
R(u_2, Q_{u_2}) & \leq H(S_1|u_1)+H(S_2|u_2) - \beta H(X_1\oplus S_1 \oplus N_{(2/3, 0, 1/3, 0)}| u_1)\\ &-(1-\beta)H(X_1\oplus S_1\oplus N_{(1/3, 0, 2/3, 0)} | u_1)-2\\
&= \beta(H(S_1|u_1)-H(X_1\oplus S_1 \oplus N_{(2/3,0,1/3,0)}|u_1))\\
&+(1-\beta)(H(S_1|u_1)-H(X_1\oplus S_1\oplus N_{(1/3,0,2/3,0)}|u_1))+H(S_2|u_2)-2
\end{align*}

\paragraph{2)}
  Let $a=0, b=2$, and $P(X_2 \oplus S_2 =0 |u_1)=p_0$, and $P(X_2 \oplus S_2 =2 |u_1)=1-p_0$. In this case $\mathbf{p}=(p_0, 1-p_0, 0, 0)$. Also,  $$\mathbf{p}=\beta (2/3, 1/3, 0, 0)+(1-\beta)(1/3,2/3,0,0),$$ where $\beta=3p_0-1$. Similar to case 1), we use Remark \ref{rem: circular conv and bilinear} and the concavity of the entropy to get,
\begin{align*}
R(u_2, Q_{u_2}) & \leq  \beta(H(S_1|u_1)-H(X_1\oplus S_1 \oplus N_{(2/3,1/3,0,0)}|u_1))\\
&+(1-\beta)(H(S_1|u_1)-H(X_1\oplus S_1\oplus N_{(1/3,2/3,0,0)}|u_1))+H(S_2|u_2)-2
\end{align*}
\end{proof}

\end{document}